\documentclass[12pt,onecolumn]{IEEEtran}
\usepackage[utf8]{inputenc}
\usepackage{amsthm}
\usepackage{amsmath}
\usepackage{amsfonts}
\usepackage{graphicx}
\usepackage{latexsym}
\usepackage{amssymb}
\usepackage{stmaryrd}
\usepackage{amscd}
\usepackage[small]{caption}
\usepackage{xcolor}
\usepackage{enumitem}

\newcommand{\eqdef}{\stackrel{\triangle}{=}}

\newcommand{\sbinom}[2]{\left[ \begin{array}{c} #1 \\ #2 \end{array} \right] }

\newcommand{\field}[1]{\mathbb{#1}}

\newcommand{\F}{\field{F}}

\newcommand{\cA}{{\cal A}}
\newcommand{\cB}{{\cal B}}

\newcommand{\cL}{{\cal L}}
\newcommand{\cS}{{\cal S}}

\newcommand{\cR}{{\cal R}}

\DeclareMathAlphabet{\mathbfsl}{OT1}{cmr}{bx}{it}
\newcommand{\uuu}{\kern-1pt\mathbfsl{u}\kern-0.5pt}
\newcommand{\vvv}{\kern-1pt\mathbfsl{v}\kern-0.5pt}

\newcommand{\myboxplus}{\kern1pt\mbox{\small$\boxplus$}}

\makeatletter \DeclareRobustCommand{\sbinom}{\genfrac[]\z@{}}
\makeatother
\newcommand{\G}[2]{\sbinom{{#1}\kern-1pt}{{#2}\kern-1pt}}
\newcommand{\Gq}[2]{\sbinom{{#1}\kern-0.25pt}{{#2}\kern-0.25pt}}

\newcommand{\bfzero}{{\boldsymbol 0}}
\newcommand{\bfone}{{\boldsymbol 1}}

\newcommand{\bfE}{{\boldsymbol E}}

\newcommand{\bfs}{{\boldsymbol s}}
\newcommand{\bfr}{{\boldsymbol r}}

\newcommand{\be}[1]{\begin{equation}\label{#1}}
\newcommand{\ee}{\end{equation}}

\newcommand{\Cref}[1]{Co\-rol\-la\-ry\,\ref{#1}}

\newtheorem{theorem}{Theorem}
\newtheorem{lemma}{Lemma}

\newtheorem{corollary}[theorem]{Corollary}

\begin{document}

\title{Efficient Algorithm for the Linear Complexity of Sequences and Some Related Consequences}

\author{\textbf{Yeow Meng Chee$^\text{x}$}, \textbf{Johan Chrisnata$^{*,+}$, \textbf{Tuvi Etzion$^*$}, \textbf{Han Mao Kiah$^\text{+}$}}\\
{\small $^\text{x}$Department of Industrial Systems Engineering and Management, National University of Singapore, Singapore}\\
{\small $^*$Computer Science Department, Technion, Israel Institute of Technology, Haifa 3200003, Israel}\\
{\small $^+$School of Physical and Mathematical Sciences, Nanyang Technological University, Singapore}\\
{\small {\it pvocym@nus.edu.sg}, {\it johan.c@cs.technion.ac.il}, {\it etzion@cs.technion.ac.il}, {\it hmkiah@ntu.edu.sg}\vspace{-0.13ex}}
}


\maketitle
\begin{abstract}
The linear complexity of a sequence $\bfs$ is one of the measures
of its predictability. It represents the smallest degree of a linear
recursion which the sequence satisfies. There are several algorithms to find the linear
complexity of a periodic sequence $\bfs$ of length $N$ (where $N$ is of some given form)
over a finite field $\F_q$ in $O(N)$ symbol field operations.
The first such algorithm is The Games-Chan Algorithm which considers binary sequences of period~$2^n$, and is
known for its extreme simplicity.
We generalize this algorithm and apply it efficiently for several families of binary sequences.
Our algorithm is very simple, it requires~$\beta N$ bit operations for a small constant $\beta$, where
$N$ is the period of the sequence. We make an analysis on the number of bit operations
required by the algorithm and compare it with previous algorithms.
In the process, the algorithm also finds the recursion for
the shortest linear feedback shift-register which generates the sequence.
Some other interesting properties related to shift-register sequences,
which might not be too surprising but generally unnoted, are
also consequences of our exposition.
\end{abstract}


\section{Introduction}

Binary sequences with good pseudorandomness and
complexity properties are widely used as keystreams in
cryptographic applications~\cite{MVV96,Rue86}. Among the measures
commonly used to measure the complexity of a sequence $\bfs$ is its
linear complexity $c(\bfs)$, defined to be the length of the shortest
linear feedback shift-register that generates $\bfs$. Sequences of low
linear complexity are fully determined via a solution of
$c(\bfs)$ linear equations if $2c(\bfs)$ consecutive terms of the sequence
are known. Hence, high linear complexity is a prerequisite for cryptographic applications,
and during the last fifty years there has been an extensive research to find the linear complexity of
sequences.

The linear complexity of a sequence $\bfs$ of length~$N$ over a finite field $\F_q$ can be determined
with the well-known Berlekamp-Massey algorithm~\cite{Ber68,Mas69} in $O(N^2)$ symbol field operations.
This algorithm was implemented during the years in various ways, e.g.~\cite{Fit95,SKHN}.
The complexity of this algorithm was improved to $O( N (\log N)^2 \log \log N)$ in~\cite{Bla97,Bla83,Bla85}.
All these algorithms considered sequences of any length over any finite field $\F_q$. But, in many applications,
only periodic sequences are considered and hence the algorithm to find the linear complexity
of such sequences, can be considerably improved.

Games and Chan~\cite{GaCh83} proposed a simple algorithm which finds the linear complexity $c$ of any
sequence $\bfs$ with period $N=2^n$. Implementation of their algorithm requires $N$ bit operations on the sequence $\bfs$
and another $n$ bit operations reduced from integer operations to compute $c$. In the following two decades
a few algorithms were suggested to generalize this algorithm to binary sequences with other periods
and periodic sequences over $\F_q$. The complexity of these algorithms for sequences
with period $N$ were kept as low as $\beta N$ for some constant $\beta$, but relatively much higher than the $N+\log N$ bit operations
required for The Games-Chan Algorithm. The generalization for sequence with period $p^n$ over $\F_{p^t}$
was given in~\cite{Din91,ImMo93}.
Xiao et al.~\cite{WXC02,XWLI} gave an algorithm to compute the linear
complexity of sequence with period $N \in \{ p^t, 2p^t\}$ over $\F_q$, when $q$ is a primitive root modulo~$p^2$.
Chen~\cite{Che05} gave an algorithm for sequences over $\F_{p^t}$ with period $\ell \cdot 2^n$,
where $2^n | p^t -1$ and $g.c.d. (\ell , p^t-1)=1$.
Chen~\cite{Che06} generalized this algorithm to determine the linear complexity
of sequences with period $\ell \cdot n$ over~$\F_{p^t}$, where $\ell | p^t - 1$
and $g.c.d.(n, p^t - 1) = 1$. The main idea in~\cite{Che06} is to reduce the calculation
of the sequence with period $\ell \cdot n$ over $\F_{p^t}$ to the calculation
of the linear complexity of $\ell$ sequences with period $n$ over $\F_{p^t}$.
The algorithms in~\cite{Che05,Che06,WXC02,XWLI} are designed for sequences over a field of odd order.
The ideas in~\cite{Che05,Che06} are generalized in~\cite{Mei08} for binary sequences.
The~\cite{Mei08} Meidl presents the most efficient algorithm for computing the linear complexities of
binary sequences of period $N=\ell \cdot 2^n$. To apply Meidl's algorithm on a sequence $\bfs$, one forms
a family of sequences of length $2^n$ from $\bfs$ and apply  The Games-Chan Algorithm to
each of these sequences. Then for specific values of $\ell$, Meidl showed that the algorithm
requires $\beta N$, where $\beta$ is a small constant. The algorithm in~\cite{Mei08}
is of interest for large $N$, a small odd integer $\ell$ such that the smallest $k$ for which $\ell$ divides $2^k - 1$ is not large.

In this paper we consider as in~\cite{Mei08} only the most important class of sequences, namely the binary sequences.
We present an algorithm which is very similar in nature to The Games-Chan Algorithm
and can be viewed as its direct generalization compared to previous algorithms, but it does
not reduce most of the computations to The Games-Chan Algorithm. Moreover, the algorithm can handle
efficiently sequences of even period as~\cite{Mei08} and also binary sequences of some odd periods,
which are not considered by previous algorithms.
The algorithm is always efficient for large $N$, but also small $N$, depending
on tradeoff between $\ell$ and $n$.
Moreover, the polynomial which generates the sequence is also computed in the algorithm,
a feature that does not exist in the algorithm of~\cite{Mei08}.
The algorithm can also compute the linear complexity of binary sequences of length
$p_1^{n_1} p_2^{n_2} ~ \cdots ~ p_t^{n_t}$, where the $p_i$'s are primes, the $n_i$'s are
positive integers, and the polynomial $\sum_{i=0}^{p_j-1} x^{i {p_j}^m}$ is primitive, for $0 \leq m \leq n_j-1$.
The algorithm requires~$\beta N$ bit operations to compute the linear complexity of a binary sequence $\bfs$ of length $N$,
where the constant $\beta$ is relatively small. Furthermore, our exposition presents also interesting,
even if not surprising, results on properties of sequences and their linear complexities.

The rest of this paper is organized as follows. In Section~\ref{sec:prelim}, we present the necessary definitions
for sequences, linear complexity, the related polynomials, and finally
The Games-Chan Algorithm is discussed. In Section~\ref{sec:general}, our general method and its algorithm is presented.
In Section~\ref{sec:power_primitive}, we consider a generalization of The Games-Chan Algorithm,
for sequences generated by a power of a primitive polynomial. In Section~\ref{sec:3power_m_2power_n}, we present
the main idea of our algorithm on sequences with period $p \cdot 2^n$, where $p$ is a prime.
The cases which will be discussed are $p=3$ and $p \equiv 1 (\text{mod}~4)$, where 2 is
a generator modulo $p$.
We compare the complexity of
our algorithm to the algorithm presented in~\cite{Mei08}. A generalization for odd period is presented in
Section~\ref{sec:odd_period}. This paper includes appendices which explain some claims
in the version for the International Symposium of Information Theory 2019.
All the computations are left for the full version of this paper.

\section{Preliminaries}
\label{sec:prelim}

Let $\bfs = \{ s_i \}_{i \geq 0}$ be an infinite binary sequence.
The sequence has a \emph{period} $N$, if $N$ is the least positive integer such that $s_i = s_{i+N}$
for each $i \geq 0$. Such a sequence is considered as a \emph{cyclic} sequence and is
be denoted by $[s_0,s_1,\ldots,s_{N-1}]$. Any periodic sequence satisfies a linear recursion
\begin{equation}
\label{eq:lin_rec}
s_{i+m} = a_1 s_{i+m-1} + \cdots + a_{m-1} s_{i+1} + a_m s_i ~, ~  i \geq 0,
\end{equation}
of order $m \leq N$, where $a_j \in \{0,1\}$ and all computations are done modulo 2.
The \emph{linear complexity} $c (\bfs)$ of $\bfs$ is defined as the least~$m$ for which~(\ref{eq:lin_rec}) holds.
Clearly ${c(\bfs) \leq N}$, for any sequence with period $N$, since $s_{i+N}=s_i$ by definition.
In terms of the \emph{shift operator}~$\bfE$, defined as $\bfE s_i = s_{i+1}$ or
$\bfE [s_0,s_1, \ldots,s_{N-1}]=[s_1,\ldots,s_{N-1},s_0]$, the linear recursion~(\ref{eq:lin_rec})
takes the form

\begin{equation}
\label{eq:lin_com}
\left(  \bfE^m + \sum_{j=1}^m a_j \bfE^{m-j} \right) s_i = 0~, ~~~~  i \geq 0.
\end{equation}
Since $m$ is the smallest such integer, it implies that ${a_m \neq 0}$.
If we define
$$
f(\bfE) \triangleq   \bfE^m + \sum_{j=1}^m a_j \bfE^{m-j} ~,
$$
then we have $f(\bfE)\bfs_i =0$ for each $i \geq 0$. Replacing the shift operator $\bfE$ by the variable $x$, we have that
$f(x) = x^m + \sum_{j=1}^m a_j x^{m-j}$, where $a_m \neq 0$. Let $c_{m-j} = a_j$
for $1 \leq j \leq m$, and hence $f(x) = x^m + \sum_{j=1}^m c_{m-j} x^{m-j}= x^m + \sum_{j=0}^{m-1} c_j x^j$,
where $c_0 \neq 0$. This implies that~(\ref{eq:lin_com}) takes the form
$$
f(\bfE)s_i = \left(  \bfE^m + \sum_{j=0}^{m-1} c_j \bfE^j \right) s_i = 0~, ~~~~  i \geq 0.
$$
If $f(x) = x^m + \sum_{j=0}^{m-1} c_j x^j$ and the sequence $\bfs = [s_0,s_1,\ldots,s_{N-1}]$ satisfies
$$
s_{i+m} = c_{m-1} s_{i+m-1} + c_{m-2} s_{i+m-2} + \cdots + c_1 s_{i+1} + c_0 s_i ~,
$$
for each $i \geq 0$, then we say that the polynomial $f(x)$ \emph{generates the sequence} $\bfs$
(or $\bfs$ is generated by $f(x)$).
It implies the following observation.
\begin{lemma}
If for $\bfs = [s_0,s_1,\ldots,s_{N-1}]$ the linear recurrence of the least degree $m$ satisfies
$$
s_{i+m} = c_{m-1} s_{i+m-1} + c_{m-2} s_{i+m-2} + \cdots + c_1 s_{i+1} + c_0 s_i ~,
$$
for each $i \geq 0$, then
$$
\left(  \bfE^m + \sum_{j=0}^{m-1} c_j \bfE^j \right) s_i = 0~,
$$
for each $i \geq 0$.
\end{lemma}

Henceforth, we denote by $\bfzero$ ($\bfone$, respectively) any sequence of any length which contains only zeroes
(ones, respectively).
\begin{corollary}
\label{cor:gen_zero}
For any nonzero polynomial $f(x)$ and any cyclic nonzero sequence $\bfs$, ${f(\bfE)\bfs =\bfzero}$, if and
only if the sequence $\bfs$ is generated by $f(x)$.
\end{corollary}
\begin{corollary}
\label{cor:before_zero}
If $\bfs$ is a nonzero cyclic sequence generated by the nonzero polynomial $f(x)^m$ and ${f(\bfE)^{m-1}\bfs =\bfr}$,
then $\bfs$ is not generated by $f(x)^{m-1}$ if and
only if the sequence $\bfr$ is a nonzero sequence generated by $f(x)$.
\end{corollary}

The sequence $\bfs$ generated by the polynomial $f(x)$, can be described in terms of
a \emph{linear feedback shift-register sequence} with a feedback function
$x_{m+1}=f(x_1,x_2,\ldots,x_m)$, where $(x_1,x_2,\ldots,x_m)$ can be any binary $m$-tuple.
As a special case, these $m$-tuples can be also those
of the sequence $\bfs$. The fundamental theory of shift-register sequences is given in~\cite{Gol67}.
A polynomial $f(x)$ is \emph{irreducible} if it has no nontrivial divisor
polynomial of a smaller degree. An irreducible polynomial of degree $k$ is called \emph{primitive}
if its nonzero roots are generators (primitive elements) of the field $\F_{2^k}$.
The nonzero cyclic sequence~$\bfs$ generated by a primitive polynomial of degree $k$ has length $2^k-1$,
and each nonzero $k$-tuple appears exactly once in a window of length $k$ in $\bfs$.
Such a sequence is called an \emph{m-sequence} for \emph{maximal length linear shift-register sequence}~\cite{Gol67}.
The period of sequences generated by other irreducible polynomials (which are not primitive) can
be also calculated by using the theory in~\cite{Gol67}. For this purpose
the following definition is given.
The \emph{exponent} of a polynomial $f(x)$ is the smallest integer $e$
such that $f(x)$ divides $x^e -1$.

There is a connection between the exponent of any polynomial $f(x)$ and the periods of the related
sequences generated by its shift-register. It does not give immediately the period of the sequences
for every polynomial $f(x)$, but it does when $f(x)$ is an irreducible polynomial.

\begin{theorem}
If $f(x)$ is an irreducible polynomial, then the nonzero sequences which it generates have
period which is equal to the exponent of $f(x)$.
\end{theorem}

Each sequence $\bfs$ can be generated by several distinct polynomials, but it appears that
the structure of the polynomials which generate $\bfs$ is determined by the polynomial of least
degree which generates $\bfs$.

For a sequence $\bfs$, the polynomial $f(\bfE)$ is a \emph{minimal zero polynomial} for $\bfs$
if $f(\bfE)$ is a polynomial of the least degree such that $f(\bfE)\bfs = \bfzero$.
In other words, in view of Corollary~\ref{cor:gen_zero}, $f(x)$ is a polynomial of least degree
which generates $\bfs$. The polynomial $f(x)$ will be called a \emph{minimal (connection) polynomial} that generates $\bfs$.

\begin{lemma}
\label{lem:gcd_pol}
If the two polynomials $f(x)$ and $g(x)$ generate the same sequence $\bfs$, then
$h(x) = g.c.d. (f(x),g(x))$ also generates $\bfs$, where $g.c.d. (\alpha(x),\beta(x))$ is the
the greatest common divisor of the polynomials $\alpha(x)$ and~$\beta(x)$.
\end{lemma}
\begin{proof}
By Corollary~\ref{cor:gen_zero}, $f(x)$ generates $\bfs$ if and only if $f(\bfE)\bfs = \bfzero$.
Similarly, $g(x)$ generates~$\bfs$ if and only if $g(\bfE)\bfs = \bfzero$.
By the Euclidian algorithm, there exists two polynomials $a(x)$ and~$b(x)$
such that $h(x) = a(x)f(x) + b(x) g(x)$. Hence, $h(\bfE) \bfs = a(\bfE)f(\bfE) \bfs + b(\bfE) g(\bfE)\bfs$,
and since $f(\bfE)\bfs = \bfzero$ and $g(\bfE)\bfs = \bfzero$, it follows that $h(\bfE) \bfs = \bfzero$.
Therefore, by Corollary~\ref{cor:gen_zero}, the polynomial $h(x)$ generates $\bfs$.
\end{proof}

\begin{corollary}
If $\bfs$ is a nonzero sequence, then it has a unique minimal zero polynomial.
\end{corollary}

The linear complexity of the sequence $\bfs$ can be computed using the greatest common divisors of
two related polynomials as follows. Let $\bfs(x)$ be the \emph{generating function} of $\bfs$ considered
as an infinite sequence, defined by
$$
\bfs(x)= s_0 + s_1 x + s_2 x^2 + s_3 x^3 + \cdots~,
$$
or by it periodic sequence
$$
\bfs^N (x) = s_0 + s_1 x + s_2 x^2 + s_3 x^3 + \cdots + s_{N-1} x^{N-1}~.
$$
The generating function $s(x)$ can be written as
$$
\bfs(x) = \frac{\bfs^N(x)}{1 - x^N} = \frac{\bfs^N(x)/g.c.d.(\bfs^N(x),1-x^N)}{(1-x^N)/g.c.d.(\bfs^N(x),1-x^N)} = \frac{g(x)}{f_\bfs(x)},
$$
where
$$
g(x)= \bfs^N(x)/g.c.d.(\bfs^N(x),1-x^N)~,
$$
$$
f_\bfs(x)= (1-x^N)/g.c.d.(\bfs^N(x),1-x^N)~.
$$
Obviously, $g.c.d.(g(x),f_\bfs(x))=1$, $\deg g(x)< \deg f_\bfs(x)$, the polynomial
$f_\bfs(\bfE)$ is the minimal zero polynomial of $\bfs$, and $\deg f_\bfs (x) = c(\bfs)$~\cite{Din91}.

Since the sequence $\bfs$ can be any binary sequence of period~$N$, it follows that computing
the linear complexity of $\bfs$ is equivalent to the computation of the degree of the
greatest common divisor of $\bfs^N (x)$ and $X^N-1$. Finding the minimal zero polynomial of $\bfs$
is equivalent to the computation of greatest common divisor of $\bfs^N (x)$ and $X^N-1$.
Hence, our algorithm can be regarded as an algorithm for the greatest common divisor of some cases.

The computation of the linear complexity in~\cite{Che05,Che06,Mei08,WXC02,XWLI} is based
on these computations of the greatest common divisor. Our method is different as it is based
on the direct definition of the linear complexity and on simple computations as done
in the well-known Games-Chan Algorithm~\cite{GaCh83}.

Binary sequences of period $2^n$ are very important among all binary sequences.
The linear complexity $c(\bfs)$ of a binary sequence $\bfs=[L ~ R]$ of period $2^n$,
where $L$ and $R$ are sequences of length $2^{n-1}$,
can be recursively computed by The Games-Chan Algorithm as follows~\cite{GaCh83}:
when $L+R= \bfzero$, then $c(\bfs)=c(L)$; otherwise we set
$c(\bfs)=2^{n-1}+c(L + R)$.
This algorithm is described in more details as follows.

\vspace{0.2cm}

\noindent
{\bf The Games-Chan Algorithm:}

\vspace{0.1cm}

The input to the algorithm is a sequence $\bfs$ of period
$2^n$. If $\bfs \neq \bfzero$, the complexity $c$ of $\bfs$
is computed recursively as follows. Initially, set $c_n=0$ and
$\cA_n=\bfs$. At a typical step of the algorithm the left half of
$\cA_m$, $L(\cA_m)=[b_0,\cdots,b_{2^{m-1}-1}]$, is added to the
right half, $R(\cA_m)=[b_{2^{m-1}},\cdots,b_{2^m-1}]$, the result
being a sequence $\cB_m$, of length $2^{m-1}$. If
$\cB_m= \bfzero$, $\cA_m$ is replaced by $\cA_{m-1}=L(\cA_m)$,
and the complexity is left unchanged, i.e., $c_{m-1}=c_m$. If $\cB_m
\neq \bfzero$, $\cA_m$ is replaced by $\cA_{m-1}=\cB_m$, and
$c_m$ is replaced by $c_{m-1}=c_m+2^{m-1}$. The complexity of $\bfs$
is given by $c(\bfs)=c_0+1$.

\vspace{0.3cm}

The number of recursive steps in the algorithm is $n$ for a sequence of period $N=2^n$,
and this is the number of integer computations (on integers of length at most $n$ bits) to obtain~$c(\bfs)$.
Since each such integer addition is for a distinct power of two, no more than a total of $n$ bit operations are required.
The algorithm also has to make at most $N$ bit computations (comparisons or additions) for
a sequence of period $N=2^n$. The reason is that in the first step~$2^{n-1}$ such
operations are required; and the number of operations is reduced by half from
one step to the following step. Thus,

\begin{theorem}
\label{thm:comp_GaCh}
The complexity in The Games-Chan Algorithm is at most $N +n$ bit operations,
for a sequence of period $N=2^n$.
\end{theorem}

\section{A General Method for any Binary Sequence}
\label{sec:general}

Let $\bfs$ be a binary sequence of period at most $N$, which implies that $(E^N -1)\bfs =\bfzero$.
Assuming that we can factorize $x^N -1$ (or equivalently $\bfE^N-1$) efficiently into irreducible factors,
we want to find the minimal zero polynomial $g(\bfE)$ of $\bfs$, i.e. the smallest factor of $\bfE^N -1$ for which
$g(\bfE) \bfs = \bfzero$ (or equivalently the factor of $x^N-1$, with the smallest degree, which generates~$\bfs$).

Let $x^N -1=q_1(x)^{\alpha_1} q_2(x)^{\alpha_2} \ldots q_t(x)^{\alpha_t}$,
where the $q_i(x)$'s are distinct irreducible polynomials and $1 \leq \alpha_i \leq 2^{\gamma_i}$
for some nonnegative integer $\gamma_i$, $1 \leq i \leq t$. We want to find the polynomial of the smallest degree
$g(x)=q_1(x)^{\delta_1}q_2(x)^{\delta_2} \ldots q_t(x)^{\delta_t}$, such that $g(\bfE)\bfs=\mathbf{0}$,
where $0 \leq \delta_i \leq 2^{\gamma_i}$, $1 \leq i \leq t$. Since at least one of
the $\delta_i$'s is nonzero, we assume w.l.o.g. that $\delta_t \neq 0$.\\
\\
For any irreducible polynomial $q(x)$, let $\cS(q(x))$ be the set of all nonzero cyclic sequences generated by
$q(x)$. For example, $\cS(x^4+x^3+x^2+x+1)=\{00011,01010,11011\}$.
\begin{lemma}
\label{lem:closeundercomposition}
Let $q(x)$ be an irreducible polynomial and $\bfr \in S(q(x))$. Then for any polynomial~$f(x)$,
such that $g.c.d.(q(x),f(x))=1$, we have that $f(\bfE)\bfr \in \cS(q(x))$.
\end{lemma}
\begin{proof}
Suppose that $\bfs = f(\bfE)\bfr$, for some sequence ${\bfr \in S(q(x))}$.
Assume first that $\bfs=\bfzero$. It implies by Corollary~\ref{cor:gen_zero} that $f(x)$ generates $\bfr$.
Since $q(x)$ also generates $\bfr$, it follows by Lemma~\ref{lem:gcd_pol}
that $g.c.d.(q(x),f(x))$ generates $\bfr$.
But $g.c.d.(q(x),f(x))=1$ and hence $g.c.d.(q(x),f(x))$ does not generates $\bfr$. Thus, $\bfs$ is a nonzero sequence.

Since $\bfr \in \cS(q(x))$ it follows that
$$
q(\bfE) \bfs = q(\bfE) (f(\bfE) \bfr) =f(\bfE) (q(\bfE)\bfr)=f(\bfE) \bfzero = \bfzero .
$$
Since $\bfs$ is a nonzero sequence, it follows that
${\bfs =f(\bfE)\bfr \in \cS(q(x))}$.
\end{proof}
Lemma~\ref{lem:closeundercomposition} leads to two interesting consequences.
\begin{corollary}
\label{cor:divisibility}
If $q(x)$ is an irreducible polynomial and $\bfr$ is a nonzero sequence generated by $q(x)$,
then for any polynomial $f(x)$, $f(\bfE) \bfr = \bfzero$ if and only if $f(x)$ is divisible by~$q(x)$.
\end{corollary}
\begin{proof}
If $f(x)$ is divisible by $q(x)$, i.e. $f(x)=h(x) q(x)$, then
$$
f(\bfE) \bfr = h(\bfE) q(\bfE)\bfr = h(\bfE) \bfzero = \bfzero .
$$
Now, suppose that $f(\bfE) \bfr = \bfzero$ and assume that $f(x)$ is not divisible
by $q(x)$. Since $q(x)$ is an irreducible polynomial, it implies that $\gcd(q(x),f(x))=1$.
Therefore, by Lemma~\ref{lem:closeundercomposition}, $f(\bfE)\bfr \in \cS(q(x))$ and hence $f(\bfE)\bfr \neq \bfzero$,
a contradiction. Thus, $f(x)$ is divisible by~$q(x)$.
\end{proof}
The second consequence from Lemma~\ref{lem:closeundercomposition} is the well known Shift and Add
property for m-sequences~\cite{Gol67}. The version given here generalizes the one which is usually used.
\begin{corollary}
Let $\bfs$ be an m-sequence generated by a polynomial $q(x)$. If
$f(x)$ is any other polynomial for which $g.c.d.(q(x),f(x))=1$, then $f(\bfE) \bfs =\bfs$.
\end{corollary}
\begin{proof}
Since $\bfs$ is an m-sequence, it follows that $q(x)$ is a primitive polynomial
and $\bfs$ is the only sequence in $\cS(q(x))$. The claim is now an immediate consequence
from Lemma~\ref{lem:closeundercomposition}.
\end{proof}
\begin{theorem}
\label{thm:long_poly}
Let $\bfs$ be a binary sequence, whose minimal polynomial is $g(x)=q_1(x)^{\delta_1}\ldots q_t(x)^{\delta_t}$,
where the $q_i$'s are distinct irreducible polynomials and $\delta_i \geq 1$, $1\leq i \leq t$. Then,
\begin{equation*}
q_1(\bfE)^{\delta_1}q_2(\bfE)^{\delta_2} \ldots q_{t-1}(\bfE)^{\delta_{t-1}}q_t(\bfE)^{\delta_t-1}\bfs \in \cS(q_t(x)).
\end{equation*}
Furthermore, for every $1 \leq i \leq t-1$, let $d_i$ be an integer such that $d_i \geq \delta_i$.  Then,
\begin{equation*}
q_1(\bfE)^{d_1}q_2(\bfE)^{d_2} \ldots q_{t-1}(\bfE)^{d_{t-1}}q_t(\bfE)^{\delta_t-1}\bfs \in \cS(q_t(x)).
\end{equation*}
\end{theorem}
\begin{proof}
If $\bfs'=q_1(\bfE)^{\delta_1}q_2(\bfE)^{\delta_2} \ldots q_{t-1}(\bfE)^{\delta_{t-1}}q_t(\bfE)^{\delta_t-1}\bfs$, then
since $g(x)$ is the minimal polynomial of $\bfs$, it follows that
$\bfs' \neq \bfzero$ and $q_t(\bfE)\bfs' =\bfzero$. This implies that $\bfs' \in \cS(q_t(x))$.
Since $\bfs' \neq \bfzero$, $q_t(\bfE)\bfs' =\bfzero$, and
${g.c.d.(q_1(x)^{d_1-\delta_1}q_2(x)^{d_2-\delta_2} \ldots q_{t-1}(x)^{d_{t-1}-\delta_{t-1}},q_t(x))=1}$,
it follows by Lemma~\ref{lem:closeundercomposition} that
$$
q_1(\bfE)^{d_1-\delta_1}q_2(\bfE)^{d_2-\delta_2} \ldots q_{t-1}(\bfE)^{d_{t-1}-\delta_{t-1}}\bfs' \in \cS(q_t(x)),
$$
and hence
\begin{equation*}
q_1(\bfE)^{d_1}q_2(\bfE)^{d_2} \ldots q_{t-1}(\bfE)^{d_{t-1}}q_t(\bfE)^{\delta_t-1}\bfs \in \cS(q_t(x)).
\end{equation*}
\end{proof}

Now, for our given sequence $\bfs$ of period $N$, we want to find the smallest $\delta_t \geq 1$ such that
$q_1(\bfE)^{\delta_1} \ldots q_t(\bfE)^{\delta_t} \bfs=\bfzero$. By Theorem~\ref{thm:long_poly}
it is sufficient to find the smallest $\delta_t \geq 1$, such that
\begin{equation*}
q_1(\bfE)^{2^{\gamma_1}}q_2(\bfE)^{2^{\gamma_2}} \ldots q_{t-1}(\bfE)^{2^{\gamma_{t-1}}}  q_t(\bfE)^{\delta_t -1} \bfs  \in \cS(q_t(x)).
\end{equation*}

The algorithm to find $\delta_t$ for $\bfs$:\\
Let $\bfr=q_1(\bfE)^{2^{\gamma_1}}q_2(\bfE)^{2^{\gamma_2}}\ldots q_{t-1}(\bfE)^{2^{\gamma_{t-1}}}\bfs$.\\
If $\bfr=\bfzero$, then $\delta_t=0$, and proceed to find
$\delta_{t-1}$. Otherwise, set $\bfr' =\bfr$ and $\delta_t=0$ ($0 < \delta_t \leq 2^{\gamma_t}$ in this case).
\\For each $j$ from $1$ to $\gamma_t$, let $\mathbf{A_{j}}=q_t^{2^{\gamma_t -j}}(\bfE)\bfr'$\\
\textbf{Case 1} : If $ \mathbf{A_j}=\bfzero$, then proceed to the next iteration.\\
\textbf{Case 2} : If $\mathbf{A_{j}} \neq \bfzero,$ then set $\bfr'=\mathbf{A_j}$, $\delta_t=\delta_t+2^{\gamma_t-j}$.\\
Set $\delta_t=\delta_t +1$ (as $\bfr' \in \cS(q_t(x))$) and either
proceed to find $\delta_{t-1}$ for $\bfs$ or
set $\bfs=q_t(\bfE)\bfr'$ and proceed to find $\delta_{t-1}$ for~$\bfs$.

\section{Powers of a Primitive Polynomial}
\label{sec:power_primitive}

The Games-Chan Algorithm can be generalized in a trivial way to sequences generated
by a power of a primitive polynomial. If the primitive polynomial $f(x)$
has degree $k$, then the period of such sequences is given by the following lemma.

\begin{lemma}
\label{lem:enumerate}
Let $f(x)$ be a primitive polynomial of degree $k$.
The polynomial $f(x)^m$ generates the all-zero sequence and $2^{(\ell-1)k-i}$
sequences whose period are $(2^k-1) \cdot 2^i$, for each  $1 \leq \ell \leq m$, where $i = \lceil \log \ell \rceil$.
\end{lemma}

Lemma~\ref{lem:enumerate} is a generalization of a similar result for $f(x)=x+1$~\cite{EtLe84}.
Sequences of period $(2^k-1) \cdot 2^i$ can have
also other minimal zero polynomials and they will be discussed later,
as the linear complexity of sequences with these periods should be computed by the algorithm
presented in this paper.
At this point we will consider the linear complexity of such sequences which are generated by a polynomial
$f(x)^m$, i.e., we would like to find the minimal~$m$ for such sequences of period $(2^k-1) \cdot 2^n$.
This will be done with an algorithm which is a straightforward generalization of The Games-Chan Algorithm.

\vspace{0.2cm}

\noindent
{\bf Power of a Primitive Polynomial (PPP) Algorithm:}

\vspace{0.1cm}

The input for the algorithm is a primitive polynomial $f(x)$ of degree $k$ (whose m-sequence is~$\bfr$),
and a nonzero sequence $\bfs$ of length $(2^k-1) 2^n$, whose minimal zero polynomial is $f(\bfE)^m$.
The output of the algorithm is $m$, such that $f(\bfE)^m$ is the minimal zero polynomial of $\bfs$.

We are looking for the minimal $m$ such that $f(\bfE)^m \bfs = \bfzero$.
By Corollary~\ref{cor:before_zero}, it implies that $f(\bfE)^{m-1} \bfs = \bfr$.
Hence, if $n=0$ then $m=1$ (this will be recognized by the formal
steps of the algorithm) and the algorithm comes to its end.

Set $\bfs_n = \bfs$ and $m_n=0$ and apply the
$n$ iterations of the algorithm; at iteration $j$, $1 \leq j \leq n$,
we have a sequence $\bfs_{n-j+1} =[\cL_{n-j+1} ~ \cR_{n-j+1}]$, whose length is $(2^k-1) \cdot 2^{n-j+1}$.
We also have a variable~$m_{n-j+1}$, where initially $m_n=0$.
In iteration $j$, we perform the following computations.

Let $f(\bfE)^{2^{n-j}} \bfs_{n-j+1} = \bfs'$, where $\bfs'$ has
length $(2^k-1) \cdot 2^{n-j}$ (possibly of a smaller period)
and distinguish between two cases:

\noindent
{\bf Case 1:} If $\bfs' \neq \bfzero$, then
set $m_{n-j}=m_{n-j+1} + 2^{n-j}$, $\bfs_{n-j} = \bfs'$, and proceed to the next iteration.\\
\noindent
{\bf Case 2:} If $\bfs' = \bfzero$, then $m-1 < m_{n-j+1} + 2^{n-j}$.
Hence, set $m_{n-j}=m_{n-j+1}$ and $\bfs_{n-j} =\cL_{n-j+1}$.

After the last iteration $\bfs_0 =\bfr$, we set $m=m_0+1$, and the algorithm comes to its end.
The linear complexity of the input sequence $\bfs$ is $km$.

The PPP algorithm can be generalized for any irreducible polynomial. In fact, there is no need
to make any modification in the algorithm. The only difference is that the ${\text{m-sequence}}$~$\bfr$ is
replaced by a set of sequences $\{ \bfr_1,\bfr_2,\ldots,\bfr_t \}$ (see Corollary~\ref{cor:before_zero}), where the period of~$\bfr_i$
is equal to the exponent of the irreducible polynomial $f(x)$. After the last iteration the sequence $\bfs_0$
which is obtained is $\bfr_i$, i.e. $f(\bfE)^{m-1} \bfs = \bfr_i$, for some $1 \leq i \leq t$. But, the algorithm
does not have to know which sequence $\bfr_j$ was obtained. Hence, $\bfr_1,\bfr_2,\ldots,\bfr_t$
or the m-sequence $\bfr$ are not part of the inputs to the algorithm. Finally, note that
the complexity of the PPP algorithm depends on the primitive polynomial.

\section{Implementation for some Periods $p \cdot 2^n$, $p$ prime}
\label{sec:3power_m_2power_n}

We are now in a position to implement our method for specific lengths of sequences,
$\ell \cdot 2^n$, where $\ell$ is odd. We will consider as examples the cases in which
$\ell=p$ is a prime such that $p =3$ or $p \equiv 1 (\text{mod}~4)$, where 2 is a generator
modulo $p$, and especially $p=5$. We will concentrate on
the number of bit operations required in these cases and compare them with the number of bit operations
required by the method of Meidl~\cite{Mei08}. The explanations for the exact steps in the implementation
are long and sometimes it is required to split the sequences into a few parts
to apply some tricky computations. The exact computations are given in Appendix A and Appendix B.
We summarize the results in these appendices with the following concluding theorem.

\begin{theorem}
Let $\bfs$ be a binary sequence of length $N$ on which the algorithm is applied.
\begin{enumerate}
\item If $N= 3 \cdot 2^n$ then $7 \cdot 2^n +2n$ bit operations are required to implement the algorithm.

\item If $N= 5 \cdot 2^n$ then $16\frac{3}{4} \cdot 2^n +2n$ bit operations are required to implement the algorithm.

\item If $N = p \cdot 2^n$, where $p \equiv 1 (\text{mod}~4)$ and 2 is a generator
modulo $p$, then $\frac{p^2+7p+7}{4} \cdot 2^{n} +2n$
bit operations are required to implement the algorithm.
\end{enumerate}
\end{theorem}

As for comparison with the method of Meidl~\cite{Mei08}, one can verify that for binary sequences of length $3 \cdot 2^n$,
the implementation requires $8 \cdot 2^n +4n$ bit operations, while for sequences of length $5 \cdot 2^n$,
the implementation requires $20 \cdot 2^n + 10n$ bit operations. For general length $\ell \cdot 2^n$ it is too
difficult to compare since it is difficult to compute the exact number of bit operations required in the algorithm of~\cite{Mei08}.
We just note that in some case the algorithm in~\cite{Mei08} is more efficient, e.g. for $N=7 \cdot 2^n$.

\section{Algorithm for Sequences with Odd Period}
\label{sec:odd_period}

When we are given a binary sequence of odd period, the situation is more complicated
and our algorithm is not efficient in many cases.
We consider two cases in which our algorithm is efficient for a binary sequence $\bfs$
of odd period $\ell$.

\noindent
{\bf Case 1:} $\ell = p^n$, $p$ prime, 2 is a generator modulo $p$.
It is well known that the polynomial
$\sum_{j=0}^{p-1} x^j$ is irreducible if and only if 2 is a generator modulo $p$.
For any given integer $i \geq 0$ the polynomial $\sum_{j=0}^{p-1} x^{j \cdot p^i}$ is also
irreducible in this case~\cite{LiNi}.

Now, one can verify that $x^{p^n}-1$ is
a multiplication of $n+1$ irreducible polynomials, where the $i$th polynomial
is $g_i(x) = \sum_{j=0}^{p-1} x^{j p^i}$, $0 \leq i \leq n-1$, i.e.
$x^{p^n} -1 = (x+1) \prod_{i=0}^{n-1} g_i(x)$.
Our algorithm can be made now simpler and more efficient.
%
%
%
%
%
The input sequence $\bfs$ has period at most $p^n$.
Initialize $\bfs'=\bfs$, $c=0$, and $f(x)=1$,
where $c$ will be the linear complexity of the
sequence $\bfs$ and $f(x)$ will be the minimal zero polynomial of $\bfs$.
Consider the $m$th iteration of the algorithm, ${1 \leq m \leq n}$.
At each iteration we have to check if $\bfs'$ whose length is $p^{n-m+1}$ has period $p^{n-m+1}$.
For this, it is sufficient to check if the first $(p-1)p^{n-m}$ bits
of $\bfr \eqdef (\bfE^{p^{n-m}}+1) \bfs'$ are equal to $\mathbf{0}$ using $(p-1)p^{n-m}$ bit operations.

\begin{enumerate}
\item If $\bfr=\mathbf{0}$, then replace $\bfs'$ with the first $p^{n-m}$ bits of $\bfs'$.

\item If $\bfr \neq \mathbf{0}$, then
set $c=c+(p-1)p^{n-m}$, $f(x)=g_{n-m}(x)f(x)$ and $\bfs'=g_{n-m}(\bfE)\bfs'$.
\end{enumerate}
After the $n$ iterations if $\bfs'=1$, then replace $c$ with $c+1$ and $f(x)$ with $(x+1)f(x)$.

Each one of the two steps in the $n$ iterations requires $(p-1)p^{n-m}$ bit operations.
The total number of bit operations required by the algorithm (omitting the additions to compute~$c$) is
$
1+\sum_{m=1}^{n}{2(p-1)p^{n-m}}=1+(2p^n-2) \leq 2p^n = 2N.
$

\noindent
{\bf Case 2:} $\ell = p_1^{n_1} p_2^{n_2} ~ \cdots ~ p_t^{n_t}$, $p_i$ prime, 2 is a generator modulo $p_i$, $n_i$ is a positive integer.
This case is solved similarly to Case~1, by considering $t$ steps, one for each prime.


\vspace{-0.07cm}

\section*{Acknowledgment}
Y. M. Chee and H. M. Kiah were supported by the Singapore Ministry of Education under
Grant MOE2015-T2-2-086.
J. Chrisnata and T. Etzion were supported by the ISF grant no. 222/19.

\section*{Appendix A}

This appendix considers the number of bit operations required to compute the linear complexity
of binary sequences whose period is $3 \cdot 2^n$.

\begin{theorem}
\label{complexityreduction3}
Let $\bfr$ be a binary cyclic sequence of the form
$\begin{bmatrix} \mathbf{X}+\mathbf{Y} & \mathbf{Y}+\mathbf{Z}& \mathbf{Z}+\mathbf{X}\end{bmatrix}$
for some binary sequences $\mathbf{X}, \mathbf{Y}, \mathbf{Z}$ of length
$2^{n-k+1}$. If $\bfs = \begin{bmatrix}
\mathbf{X} & \mathbf{Y} & \mathbf{Z}
\end{bmatrix}$ is the binary input sequence, then the following properties hold.
\begin{enumerate}
\item[(P1)] There exist binary sequences $\mathbf{X'},\mathbf{Y'},\mathbf{Z'}$ of length $2^{n-k}$,
such that $(\bfE^2+\bfE+1)^{2^{n-k}}\bfr=[\mathbf{X'}+\mathbf{Y'}, \mathbf{Y'}+\mathbf{Z'}, \mathbf{Z'}+\mathbf{X'}]$,
\item[(P2)] Obtaining the binary sequences $\mathbf{X'},\mathbf{Y'},\mathbf{Z'}$ from (P1) requires at most $3 \cdot 2^{n-k}$ bit operations,
\item[(P3)] If $(\bfE^2+\bfE+1)^{2^{n-k}}\bfr=\mathbf{0}$, then there exist binary sequences
$\mathbf{X''},\mathbf{Y''},\mathbf{Z''}$ of length $2^{n-k}$, such that
$\bfr=[\mathbf{X''}+\mathbf{Y''}, \mathbf{Y''}+\mathbf{Z''}, \mathbf{Z''}+\mathbf{X''}]$,
\item[(P4)] Obtaining the binary sequences $\mathbf{X''},\mathbf{Y''},\mathbf{Z''}$ from (P3) requires no computation.
\end{enumerate}
\end{theorem}
\begin{proof}
Suppose that $\mathbf{X}, \mathbf{Y}, \mathbf{Z}$ are binary sequences of length $2^{n-k+1},$ and $\bfr=[\mathbf{X}+\mathbf{Y}, \mathbf{Y}+\mathbf{Z}, \mathbf{Z}+\mathbf{X}]$ is a binary sequence of length $3 \cdot 2^{n-k+1}$. Let $\mathbf{X}=[\mathbf{X_1} \mathbf{X_2}],\mathbf{Y}=[\mathbf{Y_1} \mathbf{Y_2}], \mathbf{Z}=[\mathbf{Z_1} \mathbf{Z_2}]$ where $\mathbf{X_1},\mathbf{X_2} ,\mathbf{Y_1},\mathbf{Y_2},\mathbf{Z_1},\mathbf{Z_2}$ are of length $2^{n-k}$.
Then we have the following equalities for $(\bfE^2+\bfE+1)^{2^{n-k}}\bfr = (\bfE^{2^{n-k+1}}+\bfE^{2^{n-k}}+1)\bfr$
(where variables on the same column in the related matrices are supposed to be added together),
\begin{align*}
(\bfE^2+\bfE+1)^{2^{n-k}}\bfr
= \begin{bmatrix}
\mathbf{X_1}+\mathbf{Y_1} & \mathbf{X_2}+\mathbf{Y_2} & \mathbf{Y_1}+\mathbf{Z_1} & \mathbf{Y_2}+\mathbf{Z_2} & \mathbf{Z_1}+\mathbf{X_1} & \mathbf{Z_2}+\mathbf{X_2} \\
\mathbf{X_2}+\mathbf{Y_2} & \mathbf{Y_1}+\mathbf{Z_1} & \mathbf{Y_2}+\mathbf{Z_2} & \mathbf{Z_1}+\mathbf{X_1} &	\mathbf{Z_2}+\mathbf{X_2} & \mathbf{X_1}+\mathbf{Y_1} \\
\mathbf{Y_1}+\mathbf{Z_1} & \mathbf{Y_2}+\mathbf{Z_2} & \mathbf{Z_1}+\mathbf{X_1} & \mathbf{Z_2}+\mathbf{X_2} & \mathbf{X_1}+\mathbf{Y_1} & \mathbf{X_2}+\mathbf{Y_2}
\end{bmatrix}\\
=  \begin{bmatrix}
\mathbf{Z_1}+\mathbf{X_1} & \mathbf{Z_2}+\mathbf{X_2} & \mathbf{X_1}+\mathbf{Y_1} & \mathbf{X_2}+\mathbf{Y_2} & \mathbf{Y_1}+\mathbf{Z_1} & \mathbf{Y_2}+\mathbf{Z_2}\\
\mathbf{X_2}+\mathbf{Y_2} & \mathbf{Y_1}+\mathbf{Z_1} & \mathbf{Y_2}+\mathbf{Z_2} & \mathbf{Z_1}+\mathbf{X_1} &	\mathbf{Z_2}+\mathbf{X_2} & \mathbf{X_1}+\mathbf{Y_1}
\end{bmatrix}
\end{align*}
$$
=  \begin{bmatrix}
\mathbf{Z_1}+\mathbf{X_1} & \mathbf{Z_2}+\mathbf{X_2} & \mathbf{X_1}+\mathbf{Y_1}\\
\mathbf{X_2}+\mathbf{Y_2} & \mathbf{Y_1}+\mathbf{Z_1} & \mathbf{Y_2}+\mathbf{Z_2}
\end{bmatrix},
$$
.

Now, if we let \begin{align*}
\mathbf{X'}&=\mathbf{Y_2}+\mathbf{X_1},\\
\mathbf{Y'}&=\mathbf{X_2}+\mathbf{Z_1},\\
\mathbf{Z'}&=\mathbf{Y_1}+\mathbf{Z_2},
\end{align*}
then $(\bfE^2+\bfE+1)^{2^{n-k}}\mathbf{s}=[\mathbf{X'}+\mathbf{Y'}, \mathbf{Y'}+\mathbf{Z'}, \mathbf{Z'}+\mathbf{X'}]$, where $\mathbf{X'},\mathbf{Y'},\mathbf{Z'}$ are binary sequences of length $2^{n-k}$, which implies (P1).
	
Furthermore, $\mathbf{X'},\mathbf{Y'},\mathbf{Z'}$ can be computed using $3 \cdot 2^{n-k}$ bit operations since $\mathbf{X_1},\mathbf{X_2},\mathbf{Y_1},\mathbf{Y_2},\mathbf{Z_1},\mathbf{Z_2}$ are given as inputs of length $2^{n-k}$,
which implies (P2).

If $(\bfE^2+\bfE+1)^{2^{n-k}}\bfr=\mathbf{0}$, then by (P1) we have that
\begin{align*}
(\bfE^2+\bfE+1)^{2^{n-k}}\bfr
&=  \begin{bmatrix}
\mathbf{Z_1}+\mathbf{X_1} & \mathbf{Z_2}+\mathbf{X_2} &
\mathbf{X_1}+\mathbf{Y_1}\\ \mathbf{X_2}+\mathbf{Y_2} & \mathbf{Y_1}+\mathbf{Z_1} & \mathbf{Y_2}+\mathbf{Z_2}
\end{bmatrix}=\mathbf{0}, \nonumber
\end{align*}
which implies that
\begin{align*}
\mathbf{Z_1}+\mathbf{X_1} &=\mathbf{X_2}+\mathbf{Y_2},\\
\mathbf{Z_2}+\mathbf{X_2} &=	\mathbf{Y_1}+\mathbf{Z_1}, \\
\mathbf{X_1}+\mathbf{Y_1} &=\mathbf{Y_2}+\mathbf{Z_2}.
\end{align*}
Therefore,
\begin{align*}
\bfr&=\begin{bmatrix}\mathbf{X}+\mathbf{Y} & \mathbf{Y}+\mathbf{Z} & \mathbf{Z}+\mathbf{X}\end{bmatrix}\\
&=\begin{bmatrix}\mathbf{X_1}+\mathbf{Y_1}  & \mathbf{X_2}+\mathbf{Y_2} & \mathbf{Y_1}+\mathbf{Z_1}
& \mathbf{Y_2}+\mathbf{Z_2}& \mathbf{Z_1}+\mathbf{X_1} & \mathbf{Z_2}+\mathbf{X_2}\end{bmatrix}\\
&=\begin{bmatrix}\mathbf{X_1}+\mathbf{Y_1}  & \mathbf{X_1}+\mathbf{Z_1} & \mathbf{Y_1}+\mathbf{Z_1}\end{bmatrix}\\
&= \begin{bmatrix}
\mathbf{X''}+\mathbf{Y''} & \mathbf{Y''}+\mathbf{Z''}& \mathbf{Z''}+\mathbf{X''}
\end{bmatrix},
\end{align*}
where $\mathbf{X''}=\mathbf{Y_1}, \mathbf{Y''}=\mathbf{X_1}, \mathbf{Z''}=\mathbf{Z_1}$
are all given inputs which require no computation and hence (P3) and (P4) are implied.
\end{proof}

\begin{theorem}
Let $\bfs$ be a binary cyclic sequence of length $3 \cdot 2^n$ for some nonnegative integer $n$.
The computation of the linear complexity of $\bfs$ requires at most $7 \cdot 2^n +2n$ bit operations.
\end{theorem}
\begin{proof}
Since $\bfs$ is a binary cyclic sequence of length $3 \cdot 2^n$, we have that $(\bfE^{3 \cdot {2^n}}+1)\bfs=\mathbf{0}$.
The polynomial $\bfE^{3 \cdot 2^n}+1$ is factorized into irreducible polynomials, i.e., $\bfE^{3 \cdot 2^n}+1=(\bfE^2+\bfE+1)^{2^n}(\bfE+1)^{2^n}$.
To compute the linear complexity of $\bfs$, it is required to find the smallest $i$ and $j$,
such that $(\bfE^2+\bfE+1)^i(\bfE+1)^j \bfs=\mathbf{0}$. Computing the linear complexity using our algorithm can be done in two steps:
\begin{enumerate}[label=(\Alph*)]
\item Find the smallest $i$ such that $(\bfE^2+\bfE+1)^i(\bfE+1)^{2^n} \bfs=\mathbf{0}$. \label{i}
\item Find the smallest $j$ such that $(\bfE+1)^j(\bfE^2+\bfE+1)^{2^n} \bfs =\mathbf{0}$. \label{j}
\end{enumerate}
Let $\bfs=\begin{bmatrix}
\mathbf{X} & \mathbf{Y} & \mathbf{Z}
\end{bmatrix}$ be the input sequence and consider first the computations for step \ref{i}. Let
$\bfr=(\bf\bfE+1)^{2^n} \bfs=\begin{bmatrix}
\mathbf{X}+\mathbf{Y} & \mathbf{Y}+\mathbf{Z}& \mathbf{Z}+\mathbf{X}\end{bmatrix}$.
First, we initialize $\bfs'$ to be $\bfr$, and at the $k$-th iteration of PPP algorithm,
we compute $(\bfE^2+\bfE+1)^{2^{n-k}}\bfs'$. But instead of computing it,
from Theorem~\ref{complexityreduction3} (P1), we can just compute $\mathbf{X'},\mathbf{Y'},\mathbf{Z'}$ such that $(\bfE^2+\bfE+1)^{2^{n-k}}\bfs'=[\mathbf{X'}+\mathbf{Y'},
\mathbf{Y'}+\mathbf{Z'}, \mathbf{Z'}+\mathbf{X'}]$, which requires $3 \cdot 2^{n-k}$ bit operations.\\
\\
Now, to verify whether $\begin{bmatrix}
\mathbf{X'}+\mathbf{Y'}& \mathbf{Y'}+\mathbf{Z'}& \mathbf{Z'}+\mathbf{X'}
\end{bmatrix}=\mathbf{0}$ observe that $\begin{bmatrix}
\mathbf{X'}+\mathbf{Y'}& \mathbf{Y'}+\mathbf{Z'}& \mathbf{Z'}+\mathbf{X'}
\end{bmatrix}=\mathbf{0}$ if and only if
$\begin{bmatrix} \mathbf{X'}+\mathbf{Y'} & \mathbf{Y'}+\mathbf{Z'}\end{bmatrix}=\mathbf{0}$.
Therefore only $2 \cdot 2^{n-k}$ bit operations are required to compute
$\mathbf{X'}+\mathbf{Y'}$ and $\mathbf{Y'}+\mathbf{Z'}$. If the computation leads to a nonzero result, then
we proceed to the next iteration having
$\begin{bmatrix}
\mathbf{X'}+\mathbf{Y'}& \mathbf{Y'}+\mathbf{Z'}& \mathbf{Z'}+\mathbf{X'}
\end{bmatrix}$ as the new $\bfs'$. Otherwise, if the computation leads to a zero result,
then by (P3) of Theorem~\ref{complexityreduction3}, we set $\begin{bmatrix}
\mathbf{X''}+\mathbf{Y''}& \mathbf{Y''}+\mathbf{Z''}& \mathbf{Z''}+\mathbf{X''}
\end{bmatrix}$
as the new $\bfs'$ which requires no further computation. In total, for the $k$-th iteration,
the algorithm requires at most $5 \cdot 2^{n-k}$ bit operations. Therefore, to perform step \ref{i},
the algorithm (together with the integer additions) requires at most $5 \cdot 2^n+n$ bit operations.\\
\\
Next, we consider the number of bit operations required to perform step \ref{j}. Note that $(\bfE^2+\bfE+1)^{2^n} \bfs=\begin{bmatrix}
\mathbf{X}+\mathbf{Y}+\mathbf{Z}
\end{bmatrix}$, which is a binary sequence of length $2^n$. We will proceed with The Games-Chan Algorithm,
where in the first iteration we compute $(\mathbf{Y_2}+\mathbf{X_1})+(\mathbf{X_2}+\mathbf{Z_1})+(\mathbf{Y_1}+\mathbf{Z_2})=\mathbf{X'}+\mathbf{Y'}+\mathbf{Z'}$.
Since $\begin{bmatrix} \mathbf{X'}+\mathbf{Y'}\end{bmatrix}$ and $\mathbf{Z'}$
are already computed in step \ref{i} we only need $2^{n-1}$ bit operations for the
first iteration. If the computation leads to $\mathbf{0}$,
then before the second iteration, we compute the first $2^{n-1}$ bits of $\begin{bmatrix}
\mathbf{X}+\mathbf{Y}+\mathbf{Z}
\end{bmatrix}$,
namely $\begin{bmatrix}
\mathbf{X_1}+\mathbf{Y_1}+\mathbf{Z_1}
\end{bmatrix}$, which requires $2 \cdot 2^{n-1}$ bit operations. Otherwise, if the computation leads to nonzero, we proceed to the second iteration with $\begin{bmatrix}
\mathbf{X'}+\mathbf{Y'}+\mathbf{Z'}
\end{bmatrix}$ of length $2^{n-1}$ which is already computed. Hence, the first iteration
requires at most $3 \cdot 2^{n-1}$ bit operations and
at most $2^{n-2}+2^{n-3}+ \cdots +1 \leq 2^{n-1}$ bit operations are required
from the second iteration to the last iteration.
Therefore, step \ref{j} of the algorithm (together with the integer additions) requires
at most $2 \cdot 2^{n}+n$ bit operations.
\\
Thus, in total, the algorithm requires at most $7 \cdot 2^n+2n$ bit operations.
\end{proof}

\section*{Appendix B}

Let $p$ be a prime such that $2$ is a generator in the multiplicative group modulo $p$, where $p = 4\ell+1$,
for some positive integer $\ell$. Let $\bfs$ be the binary cyclic input sequence
of length $p \cdot 2^n$, which implies that ${(\bfE^{p \cdot {2^n}}+1)\bfs=\mathbf{0}}$.
The polynomial $\bfE^{p \cdot 2^n}+1$ is factorized into irreducible polynomials, i.e.,
${\bfE^{p \cdot 2^n}+1=(\bfE^{p-1}+ \cdots +\bfE^2+\bfE +1)^{2^n}(\bfE+1)^{2^n}}$.
To compute the linear complexity of $\bfs$, it is required to find the smallest $i$ and $j$,
such that
$$
{(\bfE^{p-1}+ \cdots +\bfE^2+\bfE +1)^i(\bfE+1)^j \bfs=\mathbf{0}}~.
$$
Computing the linear complexity using our algorithm can be done in two steps:
\begin{enumerate}[label=(\Alph*)]
\item Find the smallest $i$ such that $(\bfE^{p-1}+ \cdots +\bfE^2+\bfE +1)^i(\bfE+1)^{2^n} \bfs=\mathbf{0}$, \label{first}
\item Find the smallest $j$ such that $(1+\bfE)^j(\bfE^{p-1}+ \cdots +\bfE^2+\bfE +1)^{2^n} \bfs=\mathbf{0}$. \label{second}
\end{enumerate}
Similar to theorem~\ref{complexityreduction3} one can prove the following two theorems.
\begin{theorem}
Let $\bfs$ be a binary cyclic sequence of length $p \cdot 2^k$, where $p$ is a prime such that $2$ is a generator in
the multiplicative group modulo $p$, and $k \geq 0$. If $\bfs$ is the binary input sequence, then the following properties hold:
\begin{enumerate}
\item $(\bfE^{p-1}+ \cdots +\bfE^2+\bfE +1)^{2^{k-1}}(\bfE+1)^{2^k}\bfs = (\bfE+1)^{2^{k-1}} \bfs'$,
for some binary sequence $\bfs'$ of length $p \cdot 2^{k-1}$;
\item the computation of $\bfs'$ requires $p \cdot 2^{k-1}$ bit operations;
\item  $(\bfE+1)^{2^{k-1}}\bfs'=\mathbf{0}$ if and only if
the first $(p-1) \cdot 2^{k-1}$ bits of $(\bfE+1)^{2^{k-1}}\bfs'$ equal to $\mathbf{0}$.
Hence, only $(p-1) \cdot 2^{k-1}$ bit operations are required to check if
$(\bfE+1)^{2^{k-1}}\bfs'=\mathbf{0}$.
\end{enumerate}
\label{firstiteration}
\end{theorem}
\begin{theorem}
Let $\bfs$ be a binary sequence of length $p \cdot 2^k$, where $p$ is a prime such
that $2$ is a generator in the multiplicative group modulo $p$, $p=4\ell+1$ for some positive integer $\ell$.
and $k \geq 0$.
Suppose that $(\bfE^{p-1}+ \cdots +\bfE^2+\bfE +1)^{2^{k-1}}(\bfE+1)^{2^k}\bfs=\mathbf{0}$,
then given $\bfs$ as input, the following things hold,
\begin{enumerate}
\item If $\bfr=(\bfE+1)^{2^k}\bfs$, then the first half of $\bfr$ is the same as the
second half of $\bfr$, which means that $\bfr$ is a cyclic sequence of period at most $p \cdot 2^{k-1}$.
\item $(\bfE^{p-1}+ \cdots +\bfE^2+\bfE +1)^{2^{k-2}}(\bfE+1)^{2^k}\bfs$ can be computed using $(3p+1) \cdot 2^{k-2}$ bit operations.
\end{enumerate}
\label{seconditeration}
\end{theorem}
\begin{theorem}
Step \ref{first} requires at most $\frac{p^2+7p-1}{4} \cdot 2^{n}$ bit operations.
\end{theorem}
\begin{proof}
Let $\bfr=(\bfE+1)^{2^{n}} \bfs$. To perform step \ref{first} we will follow our PPP algorithm on the sequence $\bfr$.
\begin{enumerate}
\item[(1)] The first iteration of our PPP algorithm computes $(\bfE^{p-1}+ \cdots +\bfE^2+\bfE +1)^{2^{n-1}}\bfr$
to check if it is equal to the zero sequence. However, it follows from Theorem \ref{firstiteration},
that when $k=n$ the computation of $(\bfE^{p-1}+ \cdots +\bfE^2+\bfE +1)^{2^{n-1}}\bfr$ is not required.
Instead, we compute a binary sequence $\bfs'$ of length $p \cdot 2^{n-1}$,
where $(\bfE^{p-1} + \cdots +\bfE^2 +\bfE +1)^{2^{n-1}}\bfr = (\bfE+1)^{2^{n-1}} \bfs'$,
and the first $(p-1) \cdot 2^{n-1}$ bits of $(\bfE+1)^{2^{n-1}}\bfs'$, which require $(2p-1) \cdot 2^{n-1}$
bit operations for the first iteration. If $(\bfE+1)^{2^{n-1}}\bfs'$
is computed to be $\mathbf{0}$, then go to step (2); otherwise, go to step~(3).
		
\item[(2)] The second iteration of our algorithm computes $(\bfE^{p-1}+ \cdots +\bfE^2+\bfE +1)^{2^{n-2}}\bfr$.
It follows from Theorem~\ref{seconditeration}, that when $k=n$, it requires $(3p+1) \cdot 2^{n-2}$ bit operations.
If $(\bfE^{p-1}+ \cdots +\bfE^2+\bfE +1)^{2^{n-2}}\bfr$ is equal to $\mathbf{0}$,
then go to step~(4); otherwise, go to step~(5).
		
\item[(3)] The second iteration of our algorithm computes $(\bfE^{p-1}+ \cdots +\bfE^2+\bfE +1)^{2^{n-2}}(\bfE+1)^{2^{n-1}} \bfs'$.
However, it follows again from Theorem \ref{firstiteration}, that when $k=n-1$, the second
iteration of the algorithm can be done by computing a binary sequence $\bfs''$ of length $p \cdot 2^{n-2}$, where
$(\bfE^{p-1}+ \cdots +\bfE^2+\bfE +1)^{2^{n-2}}(\bfE+1)^{2^{n-1}} \bfs' = (1+\bfE)^{2^{n-2}} \bfs''$,
and the first $(p-1) \cdot 2^{n-2}$ bits of $(\bfE+1)^{2^{n-2}}\bfs''$, which require $(2p-1) \cdot 2^{n-2}$
bit operations in total for the second iteration. If $(\bfE+1)^{2^{n-2}}\bfs''$ is computed to be $\mathbf{0}$,
then go to step~(6); otherwise, go to step~(7).
		
\item[(4)] Now we have $(\bfE^{p-1}+ \cdots +\bfE^2+\bfE +1)^{2^{n-1}}\mathbf{r}=\mathbf{0}$ from step~(1)
and $(\bfE^{p-1}+ \cdots +\bfE^2+\bfE +1)^{2^{n-2}} \bfr=\mathbf{0}$ from step~(2).
It implies by Theorem \ref{seconditeration} that $\bfr$ is a sequence of period
at most $p \cdot 2^{n-1}$. Furthermore, note that
$(\bfE^{p \cdot 2^{n-2}}+1)\bfr=(\bfE^p +1)^{2^{n-2}}\bfr
=(\bfE+1)^{2^{n-2}}(\bfE^{p-1}+ \cdots +\bfE^2+\bfE +1)^{2^{n-2}}\bfr=(\bfE+1)^{2^{n-2}}\mathbf{0}=\mathbf{0}$.
Since $\bfr$ is a sequence of period at most $p \cdot 2^{n-1}$, then $(\bfE^{p \cdot 2^{n-2}}+1)\bfr=\mathbf{0}$
implies that $\bfr$ is of period at most $p \cdot 2^{n-2}$. Let $\mathbf{z}$
be the first $p \cdot 2^{n-2}$ bits of $\bfr$, which takes $p \cdot 2^{n-2}$ bit operations to compute. Go to step~(8).
		
\item[(5)] Let $\mathbf{z}$ be the first $p \cdot 2^{n-2}$ bits of
$(\bfE^{p-1}+ \cdots +\bfE^2+\bfE +1)^{2^{n-2}}\bfr$ which is already computed
in step~(2). Go to step~(8).
		
\item[(6)] From step~(3), we have $(\bfE^{p-1}+ \cdots +\bfE^2+\bfE +1)^{2^{n-2}}(\bfE+1)^{2^{n-1}} \bfs'=\mathbf{0}$.
Then it follows from Theorem \ref{seconditeration}, that when $k=n-1$,
$(\bfE+1)^{2^{n-1}}\bfs'$ is a binary sequence of period at most $p \cdot 2^{n-2}$.
Let $\mathbf{z}$ be the first $p \cdot 2^{n-2}$ bits of $(\bfE+1)^{2^{n-1}}\bfs'$
whose computation requires $p \cdot 2^{n-2}$ bit operations, and go to step~(8).
		
\item[(7)] Let $\mathbf{z}$ be the first $p \cdot 2^{n-2}$ bits of $(\bfE+1)^{2^{n-2}}\bfs''$.
Since the first $(p-1) \cdot 2^{n-2}$ bits of $(1+E)^{2^{n-2}}\bfs''$ has been computed in step~(3),
we just need to compute the last $2^{n-2}$ which requires $2^{n-2}$ bit operations. Go to step~(8).
		
\item[(8)] For any $m \geq 3$, the $m$-th iteration of our algorithm starts with a binary sequence $\mathbf{z}$
of length $p \cdot 2^{n-m+1}$, and compute the first $p \cdot 2^{n-m}$ bits of
$(\bfE^{p-1}+ \cdots +\bfE^2+\bfE +1)^{2^{n-m}}\mathbf{z}$ which requires $(p-1) \cdot p \cdot 2^{n-m}$ bit operations.
If it is nonzero then we replace $\mathbf{z}$ with the computed sequence,
otherwise we replace $\mathbf{z}$ with the first $p \cdot 2^{n-m}$ bits of $\mathbf{z}$. Repeat step~(8).
\end{enumerate}
To perform the first and the second iteration, namely step~(1) through step~(7), it requires at most
$(2p-1) \cdot 2^{n-1} +(3p+1) \cdot 2^{n-2}+p \cdot 2^{n-2}=(8p-1) \cdot 2^{n-2}$ bit operations.\\
To perform the third iteration onwards, namely step~(8),
it requires at most $(p-1) \cdot p \cdot 2^{n-2}$ bit operations.
Therefore, in total, we require at most $\frac{p^2+7p-1}{4} \cdot 2^{n}$ bit operations to perform step \ref{first}.
\end{proof}
\begin{theorem}
Step \ref{second} requires at most $2^{n+1}$ bit operations.
\end{theorem}
\begin{proof}
To perform step \ref{second}, our algorithm first computes $\bfr=(\bfE^{p-1}+ \cdots +\bfE^2+\bfE +1)^{2^{n}} \bfs$,
which requires $2^n$ bit operations. Then, we perform The Games-Chan algorithm on the sequence $\bfr$
to find its complexity, which requires $2^n$ bit operations. Therefore, in total,
at most $2 \cdot 2^n$ bit operations are required to perform step \ref{second}.
\end{proof}

\begin{corollary}
Let $\bfs$ be a binary cyclic sequence of length $N=p \cdot 2^n$, where $p=4\ell+1$, and $p$ is a generator
in the multiplicative group modulo $p$. The algorithm requires $\frac{p^2+7p+7}{4} \cdot 2^{n}$ bit operations
to compute the complexity of $\bfs$.
\end{corollary}

\begin{corollary}
$~$
\begin{enumerate}
\item If $p=5$, the algorithm requires $16\frac{3}{4} \cdot 2^n=3.35N$ bit operations,
\item if $p=13$, the algorithm requires $66\frac{3}{4} \cdot 2^n=5.14N$ bit operations,
\item if $p=29$, the algorithm requires $262\frac{3}{4} \cdot 2^n=9.1N$ bit operations.
\end{enumerate}
\end{corollary}

\end{document}